\documentclass[12pt,reqno]{amsart}
\usepackage{graphicx}
\usepackage{textcomp}
\usepackage{amsthm} 
\usepackage[top=1in, bottom=1.25in, left=1.25in, right=1.25in]{geometry}
\usepackage[colorlinks=true, citecolor=blue, urlcolor=blue ]{hyperref}
\usepackage{enumitem}
\usepackage{bm}
\usepackage{bbm}
\usepackage[stretch=10]{microtype} 
\usepackage{empheq}
\usepackage{color}
\usepackage{tcolorbox}

\usepackage[%
style=phys,%
articletitle=true,biblabel=brackets,
chaptertitle=false,pageranges=false%
]
{biblatex}

\DeclareFieldFormat{titlecase}{#1} 
\usepackage{pxfonts,txfonts}
\usepackage{tgpagella}
\usepackage[T1]{fontenc}

\newcommand{\MM}{\mathbbm{M}}

\newcommand{\HHH}{\mathcal{H}}

\newcommand{\Mod}{\mathrm{Mod}}

\DeclareMathOperator{\tr}{Tr}

\DeclareMathOperator{\rank}{rank}

\newcommand{\id}{\mathbbm{1}}
\newcommand{\idm}{\mathbbm{I}}
\newcommand{\J}{\mathbbm{J}}

\newcommand{\bra}[1]{\langle #1 |}
\newcommand{\ket}[1]{| #1 \rangle}

\newcommand{\M}{\mathcal{M}} 

\newtheorem{defn}{Definition}[section]

\newtheorem{obsr}[defn]{Observation}

\newtheorem*{note*}{Note}
\newtheorem*{rmk*}{Remark}

\newtheorem{thm}[defn]{Theorem}
\newtheorem{lem}[defn]{Lemma}
\newtheorem{prop}[defn]{Proposition}
\newtheorem{cor}[defn]{Corollary}

\theoremstyle{remark}

\newtheorem*{pf}{Proof}

  \title{Extremal Marginal States of Maximal Rank in $(d, d+m)$}
 
 \author{Indu Bala and Swapan Rana}
 \address{Indu Bala and Swapan Rana, Physics and Applied Mathematics Unit, Indian Statistical Institute, Kolkata~700108, India.}
 \email{indu.qic@gmail.com}
 \email{swapanqic@gmail.com}
 
\begin{document}
 \date{\today}

\begin{abstract}
		We study the extreme points of the convex set $\mathcal{C}(\rho_1,\rho_2)$ of bipartite quantum states with fixed marginals $\rho_1$ and $\rho_2$. We construct extreme points in $(d,\,d+m)$ dimension, of rank $d+m$, matching the highest possible value, for all $d\geq 3$, $m >  \frac{d^2-2d-2}{2}$ (when $d=2$, $m\geq 1$). This proves the existence of extremal states with relatively large rank and also covers all the known examples. We further show that, in order to analyze the extreme points of $\mathcal{C}(\rho_1,\rho_2)$, it is sufficient to study the special case $\mathcal{C}(\mathcal{D}_1,\mathcal{D}_2)$, where the marginals are diagonal. Additionally, we observe that it is sufficient to consider $d_1\leq d_2$. Thus, our results show that apart from possibly a few finite cases, for each $d_1$, the maximal rank is achieved almost all times.   
	\end{abstract}
\maketitle
	\section{Introduction}
 The combined quantum system  is associated with the Hilbert space $\HHH_1\otimes \HHH_2$. A \emph{state} $\rho$ on $\HHH$ is  a positive operator with unit trace.   $\mathcal{C}(\rho_1,\rho_2)$ be the set of all states $\rho$ on $\HHH_1 \otimes H_2$ such that the reduced density matrices with respect to $\HHH_1$ and $\HHH_2$ are given by $\rho_2$ and $\rho_1$ respectively. Since $\mathcal{C}(\rho_1,\rho_2)$ is a compact convex set so it will be in closed convex  hull of its extreme point. The necessary and sufficient condition of extremal marginal state have been  studied in \cite{LANDAU,PARTHASARATHY,Rudolph}. In particular, K.R Parthasarathy  proved that the rank of extremal  marginal state  is bounded by $\lfloor\sqrt{d_1^2+d_2^2-1}\rfloor$ \cite{PARTHASARATHY}. Later Rudolph asked the question whether 
 \[\mathcal{MR}(\rho):=\max\left\{\text{rank}(\rho):\rho\, \text{is an extreme point of the set}\, \mathcal{C}(\rho_1,\rho_2)\right\}=\left\lfloor\sqrt{d_1^2+d_2^2-1}\right\rfloor.\]
 The answer to this question depends only on the dimensions of the underlying Hilbert spaces and the choice of marginals. In the simplest case of $d_1=d_2=2$, the bound is not sharp for $\mathcal{C}\left(\frac{\idm_2}{2},\frac{\idm_2}{2}\right)$. Indeed, the extreme points of
  $\mathcal{C}\left(\frac{\idm_2}{2},\frac{\idm_2}{2}\right)$  are precisely the maximally entangled pure states \cite{PARTHASARATHY}, and hence every extreme point has rank one. Therefore, although the upper bound equals $2$, the maximal attainable rank of an extreme point in this case is only $1$. On the other hand, for several higher-dimensional cases the bound is known to be sharp. These examples indicate that the geometry of $\mathcal{C}(\rho_1,\rho_2)$ becomes considerably richer in larger dimensions.
The dimensions and marginals for which the sharpness of the bound is currently known are summarized in Table~\ref{tab:rho_pairs}.

\begin{table}[h]
	\centering
	\renewcommand{\arraystretch}{1.3}
	\setlength{\tabcolsep}{12pt}
	\begin{tabular}{c c c}
		\hline
		\textbf{$(d_1,d_2)$} & \textbf{$(\rho_1,\rho_2)$} & \textbf{Ref.} \\ 
		\hline
		
		$(2,3)$ 
		& $\left(\frac{\idm_2}{2}, \frac{\idm_3}{3}\right)$ 
		& \cite{KANMANI} \\ 
		
		$(d,d),\ d=3,4$ 
		& $\left(\frac{\idm_d}{d}, \frac{\idm_d}{d}\right)$ 
		& \cite{Ohno} \\ 
		
		$(d,d),\ d=5,9,12,5k,\ 3 \leq k \leq 14$ 
		& $\left(\frac{\idm_d}{d}, \frac{\idm_d}{d}\right)$ 
		& \cite{DEVENDRA} \\ 
		
		$(3,4)$ 
		& $\left(\frac{\idm_3}{3}, \frac{\idm_4}{4}\right)$ 
		& \cite{DEVENDRA} \\
		
		$(2,d_2),\ d_2 \geq 4$ 
		& $\left(Z, \frac{\idm_{d_2}}{d_2}\right)$ 
		& \cite{DEVENDRA} \\
		
		$(d,d+m),\ d\geq 3, m >  \frac{d^2-2d-2}{2}$ 
		& $\left(Z_1, \frac{\idm_{d+m}}{d+m}\right)$ 
		& This work \\
		
		$(2,2+m),\ m \geq  1$ 
		& $\left(Z_1, \frac{\idm_{2+m}}{2+m}\right)$ 
		& This work \\
		
		$(6,6)$ 
		& $(D,D)$ 
		& This work \\ 
		
		\hline
	\end{tabular}
	\caption{Dimensions $(d_1,d_2)$ and marginals $(\rho_1,\rho_2)$ for which sharpness of the bound is known.}
	\label{tab:rho_pairs}
\end{table}
\[\text{where}\quad Z=\frac{1}{2}\begin{pmatrix}
	1 & \frac{d_2-3}{d_2}\\
	\frac{d_2-3}{d_2} & 1
\end{pmatrix}=\frac{3}{2d_2}\idm_{2}+\frac{d_2-3}{2}\J\]

and
\[ Z_1=p \frac{\id_d}{d}+(1-p)\frac{\J}{d},\quad  p=\frac{d+1}{d+m}\]

and 
\[D=\frac{1}{9}\begin{pmatrix}
\idm_3 & 0\\
0 & 2\idm_3
\end{pmatrix}.\]
In this article, we construct an infinite family of extreme points of 
$\mathcal{C}(\rho_1,\rho_2)$ whose ranks attain the value
\[
\left\lfloor \sqrt{d_1^2+d_2^2-1}\right\rfloor .
\]
We further prove that every state in this family is separable.
Our construction therefore establishes the attainability of Parthasarathy’s bound for an infinite class of dimensions.

Another objective of this article is to simplify the study of extremal marginal states by reducing the general problem to a canonical form. We prove that, in order to study the extreme points of $\mathcal{C}(\rho_1,\rho_2)$, it is sufficient to analyze the set $\mathcal{C}(\mathcal{D}_1,\mathcal{D}_2)$, where $\mathcal{D}_1$ and $\mathcal{D}_2$ are diagonal density matrices.\\
Using the Choi--Jamiołkowski isomorphism, Rudolph \cite{Rudolph} established a one-to-one correspondence between the extreme points of
\(\mathcal{C}(\rho_1,\rho_2)\)
and those of \(\mathcal{CP}(\MM_{d_1},\MM_{d_2};\rho_1^T,\rho_2)\).
Therefore, throughout this paper we study the extremal structure via the set\\
\(\mathcal{CP}(\MM_{d_1},\MM_{d_2};\rho_1^T,\rho_2)\)
as this formulation provides a more convenient framework for analysis while fully preserving the extremality properties of the original set.
\section{Extremal Marginal State}
\begin{thm}\cite{Rudolph}\label{Ex}
Let $\Phi:\MM_{d_1}\to \MM_{d_2}$ be a completely positive map in  $\mathcal{CP}(\MM_{d_1},\MM_{d_2};\rho_1^T,\rho_2)$ then $\Phi$
is extremal in $\mathcal{CP}(\MM_{d_1},\MM_{d_2};\rho_1^T,\rho_2)$ if and only if $\Phi$ admits an expression 
\[\Phi(X)=\sum_iK_iXK_i^\dagger\] where $K_i$ are $d_1\times d_2$ matrices, satisfying the following conditions
\begin{enumerate}[label=(\roman*)]
\item $\sum_iK_i^\dagger K_i=  \rho_1^T$
\item $\sum_iK_iK_i^\dagger= \rho_2$
\item $\begin{pmatrix}
K_i^\dagger K_j & 0\\
0 & K_jK_i^\dagger 
\end{pmatrix}_{i,j} $are Linearly independent(LI)
	\end{enumerate}
\end{thm}

\begin{obsr}\label{obs}
	 $\Phi$ is extremal marginal in $\mathcal{CP}(\MM_{d_1},\MM_{d_2};\rho_1^T,\rho_2)$ if and only if $\Phi^*$ will be extremal marginal in $\mathcal{CP}(\MM_{d_2},\MM_{d_1},\rho_2,\rho_1^T)$ .
\end{obsr}

\begin{pf}
If $\Phi$  is extremal marginal in $\mathcal{CP}(\MM_{d_1},\MM_{d_2};\rho_1^T,\rho_2)$ then $\Phi$ admit an expression  given by (\ref{Ex}). Now 
\[\Phi^*(X)=\sum_i K_i^\dagger X K_i\]  and  $K_i^\dagger$ are $d_2\times d_1$ matrices will satisfy the following conditions
\begin{itemize}
	\item $\sum_iK_i K_i^\dagger=  \rho_2$
	\item $\sum_iK_i^\dagger K_i= \rho_1^T$
	\item $\begin{pmatrix}
		K_i K_j^\dagger & 0\\
		0 & K_j^\dagger K_i
	\end{pmatrix}_{i,j}$will be LI as
	 $\begin{pmatrix}
	K_i^\dagger K_j & 0\\
	0 & K_jK_i^\dagger 
	\end{pmatrix}_{i,j}$are LI.
\end{itemize}
Hence $\Phi^*$ will extremal marginal in $\mathcal{CP}(\MM_{d_2},\MM_{d_1},\rho_2,\rho_1^T)$ , converse part will follow by same argument.
	\end{pf}
	
\begin{note*}\label{note}
To study the extreme points of $\mathcal{CP}(\MM_{d_1},\MM_{d_2};\rho_1^T,\rho_2)$, it is enough to study the case where both $\rho_1$ and $\rho_2$ are invertible\cite{DEVENDRA}. 
\end{note*}

\begin{obsr}\label{D1,D2}
To study the extreme points of $\mathcal{CP}(\MM_{d_1},\MM_{d_2};\rho_1^T,\rho_2)$, it is enough to study the case where both $\rho_1$ and $\rho_2$ are full rank diagonal .
\end{obsr}
\textbf{Proof}: If $\Phi$ is extremal marginal in  $\mathcal{CP}(\MM_{d_1},\MM_{d_2};\rho_1^T,\rho_2)$ then there will exist $K_i$ such that $K_i$ satisfy the condition  of theorem  $(\ref{Ex})$.
Consider the unitary U and W such that $U\rho_1^TU^\dagger=D_1$ and $V\rho_2V^\dagger=D_2$. 
Now construct  a new channel $\Psi$ 
\[\Psi(X)=\sum_i VK_iU^\dagger XUK_i^\dagger V^\dagger=\sum_i A_iXA_i^\dagger\]
where $A_i=VK_iU^\dagger.$ Now we can check easily
\begin{itemize}
\item $\sum_iA_i^\dagger A_i=D_1$
\item $\sum_i A_iA_i^\dagger=D_2$
\item $\begin{pmatrix}
A_i^\dagger A_j & 0\\
0 & A_j A_i^\dagger
\end{pmatrix}=\begin{pmatrix}
UK_i^\dagger K_j U^\dagger& 0\\
0 & VK_j K_i^\dagger V^\dagger
\end{pmatrix}=\begin{pmatrix}
U &0 \\
0 & V
\end{pmatrix}\begin{pmatrix}
	K_i^\dagger K_j& 0\\
	0 & K_j K_i^\dagger
\end{pmatrix}\begin{pmatrix}
U^\dagger &0 \\
0 & V^\dagger
\end{pmatrix}$\\
is linear Independent as  a  unitary transformation will send LI to LI. 
Hence  $\Psi$ will be extremal marginal in $CP(\MM_{d_1},\MM_{d_2};D_1,D_2).$ 
\end{itemize}

\begin{rmk*}
 Studying the extreme points of  $\mathcal{CP}(\MM_{d_1},\MM_{d_2};\rho_1^T,\rho_2)$, is equivalent to considering the case where  $\rho_1$ and $\rho_2$ are full rank diagonal and $d_1\leq d_2$ as follows from observation \ref{obs} and observation \ref{D1,D2}.
\end{rmk*}

\section{Extremal marginal states in \(\left(d,d+m\right)\) }

\begin{thm}\label{$(d+m,d)$}
There exist an extreme point of the set  $\mathcal{CP}(\MM_{d},\MM_{d+m};Z,\frac{\idm_{d+m}}{d+m})$ with choi rank $d+m$ $\forall m\geq 1\,\text{and}\,d\geq2$, where $
Z= p\frac{\idm_d}{d}+(1-p)\frac{\J}{d}, \,\text{and}\,\J $ is the $d-$ dimensional matrix whose all entries are 1 and $p=\frac{d+1}{d+m}$.
\end{thm}
\begin{pf}
Let	$\Phi:\M_{d}\to \M_{d+m}$ be a CP map defined by
\begin{equation}\label{C.P.E}
\Phi(X)=\frac{1}{d(d+m)}\sum_{i=1}^{d+m}V_i^\dagger XV_i
\end{equation}
where 

\begin{equation}\label{K.r.1}
	V_i=\sum_{k=1}^{d}\ket {e^d_k}\bra{e^{d+m}_{\Mod(k+i-2,d+1)+1}}\quad \forall i=1,2...d+1
	\end{equation}
	
\begin{equation}\label{K.r.2}
V_i=\ket{\J}\bra{e^{d+m}_i}\quad\forall i=d+2,.....d+m 
\end{equation}

\textbf{Claim}: $\Phi$ is extremal in $\mathcal{CP}\left( \MM_{d},\MM_{d+m}; Z, \frac{\idm_{d+m}}{d+m}\right)$.

We can check 
\[\frac{1}{d(d+m)}\sum_{i=1}^{d+m}V_i V_i^{\dagger}=\frac{I_d}{d}+\frac{m-1}{d(d+m)}(\J-I_d)\]

\begin{equation}\label{rho1}
	=p \frac{I_d}{d}+(1-p)\frac{\J}{d},\,\text{where}\,  p=\frac{d+1}{d+m}
\end{equation}
and
\begin{equation}\label{rho2}
\frac{1}{d(d+m)}\sum_{i=1}^{d+m}V_i^\dagger V_i= \frac{I_{d+m}}{d+m}
\end{equation}

To show that the 3rd condition of Theorem~\ref{Ex} is satisfied, we need to find out 

\[M_{i,j}=\begin{pmatrix}
	V_i^\dagger V_j & 0\\
	0 & V_jV_i^\dagger
\end{pmatrix}^{d+m}_{i,j=1}\] 
Let \(A\) be the matrix whose \(((i-1)(d+m)+j)^{\text{th}}\) column is \(\ket{M_{i,j}}\) (We are doing  matrix vectorization column by column). So $A$ is  of order $d+m\times (d+m)^2.$ To show that the vectors \(\ket{M_{i,j}}\) are linearly independent, we will show that \(A A^T = M\) is of full rank.  
A detailed calculation of \(M\) is given in the Appendix in (\ref{AAT=M}). After simplification we will get

\begin{equation*}
\begin{aligned}
		M &=\sum_{x,y=1}^{d+m}(E_{x,x}\otimes E_{y,y}+E_{x,y}\otimes E_{x,y})+(d^2-1)\sum_{r,s=d+2}^{d+m}(E_{r,s}\otimes E_{r,s}+E_{r,r}\otimes E_{s,s})\,+\\
		&2(d-1)\sum_{i,j=1}^{d+1}E_{i,j}\otimes S^{Mod[i-j-1,d+1]+1}
		+(d-1)\sum_{r=d+2}^{d+m}\sum_{i,j=1}^{d+1}(E_{i,r}\otimes E_{j,r}
		+E_{r,i}\otimes E_{r,j}
		+E_{i,j}\otimes E_{r,r}+E_{r,r}\otimes E_{i,j})
\end{aligned}	
\end{equation*}\vspace{0.3cm}
\begin{equation*}
\begin{aligned}
&\quad\quad\,=I_{d+m}\otimes I_{d+m}+\ket{I_{d+m}}\bra{I_{d+m}}+(d^2-1)\ket{0\oplus I_{m-1}}\bra{0\oplus I_{m-1}}+(d^2-1)\left(0\oplus I_{m-1}\right)\otimes \left(0\oplus I_{m-1}\right)\\
&\quad\quad\, +2(d-1)\sum_{i,j=1}^{d+1}E_{i,j}\otimes S^{Mod[i-j-1,d+1]+1}\,+(d-1)\left(\ket{\J\oplus0}\bra{0 \oplus I_{m-1}}+\ket{0\oplus I_{m-1}}\bra{\J \oplus 0}\right)+\\
&\quad\quad\,\quad +(d-1)\left(\left(\J\oplus 0\right)\otimes\left(0\oplus I_{m-1}\right)+\left(0\oplus I_{m-1}\right)\otimes \left(\J\oplus 0 \right)\right)
\end{aligned}	
\end{equation*}
  Hence 
\begin{equation}\label{M}
\begin{aligned}
M&= I_{d+m}\otimes I_{d+m}+\ket{I_{d+m}}\bra{I_{d+m}}+2(d-1)\ket{0\oplus I_{m-1}}\bra{0\oplus I_{m-1}}\\ 
&+2(d-1)\left(0\oplus I_{m-1}\right)\otimes \left(0\oplus I_{m-1}\right)+
+(d-1)\ket{J\oplus I_{m-1}}\bra{J\oplus I_{m-1}}\\
&+(d-1)\left(J\oplus I_{m-1}\right)\otimes\left(J\oplus I_{m-1}\right)+
2(d-1)\sum_{i,j=1}^{d+1}E_{i,j}\otimes S^{Mod[i-j-1,d+1]+1}
\end{aligned}
\end{equation}
where $ S=\sum_{k=1}^{d+1}\ket{e_{\Mod[k,d+1]+1}}\bra{k}$.\\

 \(Y=2(d-1)\sum_{i,j=1}^{d+1}E_{i,j}\otimes S^{Mod[i-j-1,d+1]+1}\)(last term of equation (\ref{M})), is positive semidefinite as we can write 
\[Y=\begin{pmatrix}
	\idm  \\  
	S\\
	S^2\\
	.\\
	.\\
	S^d\\
	0\\
	.\\
	.\\
	0
\end{pmatrix}\begin{pmatrix}
	\idm & S^\dagger & S^{2^\dagger}  &.  &.  &S^{d^\dagger} & 0 & . &. &0
\end{pmatrix}\]
Since in equation (\ref{M}) each term  is positive semidefinite after  $I_{d+m}\otimes I_{d+m}$, it follows that $M$ is full rank, thereby completing the proof.

\end{pf}

\begin{cor}
	If $d\geq 2$, m$\geq 1$ then $\mathcal{MR}(Z,\frac{I_{d+m}}{d+m})\geq d+m$. Moreover, $\mathcal{MR}(Z,\frac{I_{d+m}}{d+m})=d+m$ $\forall$ $m > \frac{d^2-2d-2}{2}.$
\end{cor}
\begin{proof}
 Since
  \[x-1< \lfloor{x}\rfloor\leq x\]
and  \[\mathcal{MR}\left(Z,\frac{I_{d+m}}{d+m}\right)= \lfloor \sqrt{d^2+(d+m)^2-1}\rfloor=d+m\]
then
\[\sqrt{d^2+(d+m)^2-1}-1<d+m\]
\[\implies d^2+(d+m)^2-1< (d+m)^2+2(d+m)+1\]
\[\implies d^2< 2(d+m)+2\]
\[\implies m > \frac{d^2-2d-2}{2}\]
\end{proof}
\begin{cor}
	Theorem~\ref{$(d+m,d)$} covers both Theorem~5.8 and Theorem~5.2 in \cite{DEVENDRA}. 
\end{cor}
\begin{proof}
To prove Theorem 5.8 and Theorem 5.2 in \cite{DEVENDRA}, set m=1 and d=2, respectively, in (\ref{C.P.E}).

\end{proof}
\begin{prop}
Completely positive map definded in Eq.~(\ref{C.P.E}) is   entanglement breaking.
\end{prop}

\begin{proof}
 Let $\rho$ be the Choi state corresponding to $(\ref{C.P.E})$. We  shows in Appendix in lemma \ref{$rho^T$} that
 $\rho^{T_1} \geq 0$. Since the Kraus rank coincides with the Choi rank, we have
 \[
 \rank(\rho)=d+m=d_2 .
 \]
 For states satisfying $\rank(\rho)\leq d_2$, the PPT condition is necessary and sufficient for separability \cite{Horodecki}. Therefore, the map defined in $(\ref{C.P.E})$ is entanglement breaking.
\end{proof}

\begin{cor}
$\Phi$	defined in (\ref{C.P.E}) have entanglement breaking rank $= d+m$.
\end{cor}

\begin{proof}
Since the completely positive map defined in (\ref{C.P.E}) is unital and entanglement breaking with Kraus rank $d+m=d_2$, it follows from Theorem 5.2 in \cite{Bhat}   that \[EB-\rank(\Phi)=d+m\]
\end{proof}

\section{Extremal Marginal state in $(6k,6k, k\geq 1)$}
\begin{thm}[Theorem 2.6, \cite{Ohno}]\label{rank-4}
Let $\Phi:\MM_3\to\MM_3$ be a CP map defined by 
\begin{align}\label{CP-extreme-rank-4-2}
	\Phi(X)=\frac{1}{4}\sum_{j=1}^{4}B_jXB_j^\dagger
\end{align}
for all $X\in\MM_{3}$, where
\small
\begin{equation}\label{eq-extr-rank-4}
	\begin{cases}
		B_1&=E_{11}, \\
		B_2&=E_{12}+\sqrt{2}E_{23},\\
		B_3&=\sqrt{2}E_{21}+\sqrt{3}E_{32},\\
		B_4&=E_{31}+\sqrt{2}E_{13}.  
	\end{cases}
\end{equation}
\normalsize
Then $\Phi$ is an extreme point of the set $\mathcal{CP}(\MM_3,\MM_3; \frac{\idm_{3}}{3},\frac{\idm_{3}}{3})$ with $CR(\Phi)=4$.
\end{thm}

\begin{thm}[Theorem 2.8, \cite{Ohno}]\label{CP-extrem-rank-d}
	Given $d\geq 3$, define the map $\Phi(X):\MM_{d}\to\MM_{d}$ by
	\begin{align}\label{extreme-rank-d}
		\Phi=\sum_{j=1}^{d}V_jXV_j^\dagger,
	\end{align}
	where
	\begin{equation}\label{eq-extreme-rank-lowerbound}
		\begin{cases}
			V_1&=\sqrt{\frac{d-2}{d-1}}\sum\limits_{j=2}^{d}E_{jj},\\
			V_k&=\frac{1}{\sqrt{d-1}}(E_{1k}+E_{k1}), \text{ for all } 2\leq k\leq d.
		\end{cases}
	\end{equation}
	Then $\Phi\in\mathcal{CP}(\MM_d,\MM_d; I_{d}, I_{d})$ is an extreme point of the set $\mathcal{CP}(\MM{d},\MM{d}; \idm_{d})$ with Choi $\rank\Phi=d$
\end{thm}

  \begin{thm}[Theorem 3.6, \cite{DEVENDRA}\label{CP-tensor-extreme}]
	Let $\Phi\in\mathcal{CP}[\MM_d,\MM_d; A_1, A_2]\subseteq \mathcal{CP}[\MM_{d},\MM_{d}; A_2]$, which has a minimal Kraus decomposition consisting of Hermitian Kraus operators. Let $\Psi\in\mathcal{CP}[\MM_{d_1},\MM_{d_2}; B_1, B_2]$. If $\Phi$ and $\Psi$ are both extreme points of the sets $\mathcal{CP}[\MM_{d},\MM_{d}; A_2]$ and $\mathcal{CP}[\MM_{d_1},\MM_{d_2}; B_1, B_2]$, respectively, then the tensor product $\Phi\otimes \Psi$ is also an extreme point of the set $\mathcal{CP}[\MM_{d}\otimes\MM_{d_1},\MM_{d}\otimes\MM_{d_2}; A_1\otimes B_1, A_2\otimes B_2]$.
\end{thm}

\begin{thm}
	There exists an extreme point of the set $\mathcal{CP}(\MM_2,\MM_2;\sigma,\sigma)$ of Kraus rank 2 where 
	\[\sigma=\begin{pmatrix}
		\frac{1}{3} & 0\\
		0 & \frac{2}{3}
	\end{pmatrix}.\]	
\end{thm}

\begin{proof}

	Let  $\Phi:\MM_2\to \MM_2$ be a CP map defined by
	
	\begin{equation}\label{CP-rank-2}
	\Phi(X)=\frac{1}{2}\sum_{i=1}^2 A_iXA_i^\dagger
	\end{equation}
	for all $X\in \MM_2$, where

	\begin{equation}\label{eq-rank2-Herm}
		A_1=\begin{pmatrix}
			\frac{1}{\sqrt{3}} & 0 \\
			0 & 1 
		\end{pmatrix},
		\quad A_2=\begin{pmatrix}
			0 & \frac{1}{\sqrt{3}} \\
			\frac{1}{\sqrt{3}} & 0
		\end{pmatrix}
	\end{equation}
	
	The pairs of sets $\{A_i^\dagger A_j: 1\leq i\leq j\leq 2\}$ are LI and $\sum A_i^\dagger A_i=\frac{\id}{2}\sum_{i=1}^2A_i^\dagger A_i=\sigma $. So $\Phi$ is extremal in $\mathcal{CP}(\MM_2,\MM_2,\sigma).$
\end{proof}
\begin{thm}\label{rank-8}
There exists an extreme point of the set $\mathcal{CP}(\MM_{6},\MM_{6}; D, D)$ with Choi-rank $8$
where $D=\sigma\otimes \frac{\idm_3}{3}$.
\end{thm}

\begin{proof}
Consider the CP maps $\Phi_1:\MM_2\to \MM_2$ and  $\Phi_2:\MM_3\to \MM_3$ given by equations $(\ref{CP-rank-2})$ and (\refeq{CP-extreme-rank-4-2}) respectively
\[	\Phi_1(X)=\frac{1}{2}\sum_{i=1}^2 A_iXA_i^\dagger\]
\[	\Phi_2(X)=\frac{1}{4}\sum_{i=1}^4 B_iXB_i^\dagger\]
where  the matrices $A_i$ and $B_i$ are given by  the equations (\ref{eq-rank2-Herm}) and (\ref{eq-extr-rank-4}) respectively.
$\Phi_1$ is an extreme point of the set $\mathcal{CP}(\MM_2,\MM_2; \sigma)$ with Choi-rank $(\Phi_1)=2$, which has a minimal Kraus decomposition consisting of Hermitian Kraus operators  $\{A_i\}_{i=1}^2$.Also, by Theorem \ref{rank-4}, $\Phi_2$ is an extreme point of the set $\mathcal{CP}(\MM_3,\MM_3; \frac{\idm_{3}}{3},\frac{\idm_{3}}{3})$.

Define a new CP map 

$\Phi:\MM_6\to \MM_6$
\begin{equation}\label{CP-rank-8}
	\Phi(X)=\frac{1}{8}\sum_{i;j=1}^{2;4}\left(A_i\otimes B_j\right)X\left(A_i\otimes B_j\right)^\dagger
	\end{equation}
$\Phi$ have choi rank 8.
Then $\Phi$ is an extreme point of the set $ \mathcal{CP}(\MM_6,\MM_6,D,D)$ follows from the Theorem \ref{CP-tensor-extreme}.
\end{proof}

\begin{cor}
$\mathcal{MR}(D,D)=8=\lfloor{ \sqrt{6^2+6^2-1}}\rfloor$
\end{cor}

\begin{thm}
	There exists an extreme point of the set $\mathcal{CP}(\MM_{6k},\MM_{6k}; D_1, D_1)$ with Choi-rank $8k$
	where $D_1=\sigma\otimes \frac{\idm_3}{3}\otimes\frac{\idm_k}{k}$ where $k\geq 3$.
\end{thm}
\begin{proof}
Consider the CP maps $\Phi_1:\MM_k\to \MM_k$ and  $\Phi_2:\MM_6\to \MM_6$ given by Eqs.~(\ref{extreme-rank-d}) and (\ref{CP-rank-8}) respectively
\[	\Phi_1(X)=\frac{1}{k}\sum_{i=1}^k V_iXV_i^\dagger\]
\[	\Phi_2(X)=\frac{1}{8}\sum_{i;j=1}^{2;4}\left(A_i\otimes B_j\right)X\left(A_i\otimes B_j\right)^\dagger\]
$\Phi_1$ is an extreme point of the set $\mathcal{CP}(\MM_k,\MM_k; \sigma)$ with Choi-rank $(\Phi_1)=k$, which has a minimal Kraus decomposition consisting of Hermitian Kraus operators  $\{V_i\}_{i=1}^k$.Also, by Theorem \ref{rank-8}, $\Phi_2$ is an extreme point of the set $\mathcal{CP}(\MM_6,\MM_6; D,D)$.	
\end{proof}
Now construct a  new CP map $\Phi:\MM_{6k}\to \MM_{6k}$
\begin{equation}
	\Phi(X)=\frac{1}{8k}\sum_{i,j,k=1}^{2;4;k}\left(V_k\otimes A_i\otimes B_j \right)X\left(V_k\otimes A_i\otimes B_j\right)^\dagger
\end{equation}
$\Phi$ have Choi rank $8k$.Then $\Phi$ is an extreme point of the set $ \mathcal{CP}(\MM_6,\MM_6,D_1,D_1)$ follows from the Theorem \ref{CP-tensor-extreme}.
\begin{cor}
$\mathcal{MR}(D_1,D_1)>8k$ when $k\geq 3.$	
\end{cor}

\section{Appendix}\label{appendix}

\subsection{Calculation for $AA^T=M$}\label{AAT=M}

\begin{note*}
	If $M_i=\begin{pmatrix}
		A^1_i & 0\\
		0 & A^2_i
	\end{pmatrix}$ and A is the matrix whose $i^{th}$ columns  is $\ket{M_i}$ then 
	$A^\dagger A=\sum_{i,j}\left(\tr (A_i^{1^\dagger} A^1_j)+\tr(A_i^{2^\dagger} A^2_j)\right)E_{i,j}.$
\end{note*}
Now we have 
\begin{align}
	\begin{cases}
		V_i &= \sum_{k=1}^{d} \ket{e^d_k} \bra{e^{d+m}_{\text{mod}(k+i-2, d+1)+1}} \quad i = 1, 2, \dots, d+1 \\
	    W_i &= \ket{\J} \bra{e^{d+m}_i} \quad i = d+2, \dots, d+m
	\end{cases}
\end{align}
\[M_{i,j}=\begin{pmatrix}
	V_i^\dagger V_j & 0\\
	0 & V_jV_i^\dagger
\end{pmatrix}^{d+m}_{i,j=1}\] 
\small
\begin{align*}
\begin{cases}
	\vspace{0.3 cm}
\begin{pmatrix}
V_i^\dagger V_j & 0\\
	0 & V_jV_i^\dagger
\end{pmatrix}
^{d+1}_{i,j=1}& =
\begin{pmatrix}
	\sum^d_{k=1}\ket{e^{d+m}_{\Mod[k+i-2,d+1]+1}}\bra{e^{d+m}_{\Mod[k+j-2,d+1]+1}} & 0\\
	0 &\sum^d_{k,l=1}\delta_{\Mod[k + j - l-1, d + 1] + 1,i} \ket{e_k}\bra{e_l}
\end{pmatrix}^{d+1}_{i,j=1}\\
\vspace{0.3 cm}
\begin{pmatrix}
	W_i^\dagger W_j & 0\\
	0 & W_jW_i^\dagger
\end{pmatrix}^{d+m}_{i,j=d+2}&
=\begin{pmatrix}
	d\ket{ e^{d+m}_i}\bra{ e^{d+m}_j}& 0\\
	0 & \delta_{i,j}\ket{J}\bra{J}
\end{pmatrix}^{d+m}_{i,j=d+2}\\
\vspace{0.3cm}
\begin{pmatrix}
	V_i^\dagger W_j & 0\\
	0 & W_jV_i^\dagger
\end{pmatrix}_{i=1,j=d+2}^{d+1,d+m} &=
\begin{pmatrix}
	\sum_{k=1}^d\ket{e^{d+m}_{\Mod[k+i-2,d+1]+1}}\bra{e^{d+m}_{j}} & 0\\
	0 & 0
\end{pmatrix}^{d+1,d+m}_{i,j=1,d+2}\\
\vspace{0.3 cm}
\begin{pmatrix}
	W_i^\dagger V_j & 0\\
	0 & V_jW_i^\dagger
\end{pmatrix}^{d+m,d+1}_{i=d+2,j=1}&
=\begin{pmatrix}
	\sum_{k=1}^d\ket{e^{d+m}_i}\bra{e^{d+m}_{\Mod[k+j-2,d+1]+1}}& 0\\
	0 & 0
\end{pmatrix}^{d+m,d+1}_{i=d+2,j=1}
\end{cases}
\end{align*}
\normalsize
	In this way total matrices will be
	$(d+1)^2+(m-1)^2+2(d+1)(m-1=(d+m)^2$\\
	Let \(A\) is the matrix whose $((i-1)(d+m)+j)^{th}$ column is $\ket{M_{i,j}}$. For $AA^T$ we need to calculate.

\begin{align*}
\begin{cases}
\left(\tr(A^{1^\dagger}_{i,j} A^1_{r,s})+\tr(A^{2^\dagger}_{i,j} A^2_{r,s})\right)_{i,j=1}^{d+1}
&= \sum_{k,l=1}^d\left(\delta_{\Mod[k+i-2,d+1]+1,\Mod[l+r-2,d+1]+1}\right)\\&\left(\delta_{\Mod[k+j-2,d+1]+1,\Mod[l+s-2,d+1]+1}\right)+\\&\delta_{\Mod[k+j-l-1,d+1]+1,i}\delta_{\Mod[k+s-l-1,d+1]+1,r}\\
\vspace{0.4 cm}
\left(\tr(A^{1^\dagger}_{i,j} B^1_{r,s})+\tr(A^{2^\dagger}_{i,j} B^2_{r,s})\right)_{\{i,j=1;r,s=d+1\}}^{d+1,d+m}&=\delta_{r,s}\sum_{k,l=1}^d\delta_{\Mod[k+j-l-1,d+1]+1,i}\\
\vspace{0.4 cm}
\left(\tr(B^{1^\dagger}_{i,j} A^1_{r,s})+\tr(B^{2^\dagger}_{i,j} A^2_{r,s})\right)_{\{i,j=d+1;r,s=1\}}^{d+m,d+1}&=\delta_{i,j}\sum_{k,l=1}^d\delta_{\Mod[k+s-l-1,d+1]+1},r\\
\vspace{0.4 cm}
\left(\tr(B^{1^\dagger}_{i,j} B^1_{r,s})+\tr(B^{2^\dagger}_{i,j} B^2_{r,s})\right)_{\{i,j=d+2;r,s=d+2\}}^{d+m,d+m}&=d^2\delta_{i,r}\delta_{j,s}+d^2\delta_{i,j}\delta_{r,s}\\
\vspace{0.4 cm}
\left(\tr(C^{1^\dagger}_{i,j} C^1_{r,s})+\tr(C^{2^\dagger}_{i,j} C^2_{r,s})\right)_{\{i,j=1,d+2;r,s=1,d+2\}}^{\{d+1,d+m;d+1,d+m\}}&=\delta_{j,s}\sum_{k,l=1}^d\delta_{\Mod[k+i-2,d+1]+1,\Mod[l+r-2,d+1]+1}\\
\left(\tr(D^{1^\dagger}_{i,j} D^1_{r,s})+\tr(D^{2^\dagger}_{i,j} D^2_{r,s})\right)_{\{i,j=d+2,1;r,s=d+2,1\}}^{\{d+m,d+1;d+m,d+1\}}&=\delta_{i,r}\sum_{k,l=1}^d\delta_{\Mod[k+j-2,d+1]+1,\Mod[l+s-2,d+1]+1}
\end{cases}
\end{align*}

\begin{equation*}
\begin{aligned}
&M=\sum_{i,j,r,s=1}^{d+1}\sum_{k,l=1}^d(\delta_{\Mod[k+i-2,d+1]+1,\Mod[l+r-2,d+1]+1}\delta_{\Mod[k+j-2,d+1]+1,\Mod[l+s-2,d+1]+1}\\
&\quad+\delta_{\Mod[k+j-l-1,d+1]+1,i}\delta_{\Mod[k+s-l-1,d+1]+1,r})E'_{(i-1)(d+m)+j,(r-1)(d+m)+s}\\ &\quad+\sum_{i,j=1;r=d+2}^{d+1;d+m}\sum_{k,l=1}^d\delta_{\Mod[k+j-l-1,d+1]+1,i}E'_{(i-1)(d+m)+j,(r-1)(d+m)+r}\\
&\quad+\sum_{i=d+2;r,s=1}^{d+m;d+1}\sum_{k,l=1}^d\delta_{\Mod[k+s-l-1,d+1]+1,r} E'_{(i-1)(d+m)+i,(r-1)(d+m)+s}\\
&\quad+\sum_{i,j=d+2}^{d+m}d^2E'_{(i-1)(d+m)+j,(i-1)(d+m)+j}+d^2\sum_{i,r=d+2}^{d+m}E'_{(i-1)(d+m)+i,(r-1)(d+m)+r}\\
&\quad+\sum_{i,r=1;j=d+2}^{d+1;d+m}\sum_{k,l=1}^d\delta_{\Mod[k+i-2,d+1]+1,\Mod[l+r-2,d+1]+1 }E'_{(i-1)(d+m)+j,(r-1)(d+m)+j}\\
&\quad+\sum_{i=d+2;j,s=1}^{d+m;d+1}\sum_{k,l=1}^d\delta_{\Mod[k+j-2,d+1]+1,\Mod[l+s-2,d+1]+1}E'_{(i-1)(d+m)+j,(i-1)(d+m)+s}
\end{aligned}
\end{equation*}

\begin{equation*}
\begin{aligned}
&=d\sum_{i,j=1}^{d+1}E'_{(i-1)(d+m)+j,(i-1)(d+m)+j}+(d-1)\sum_{i,j,r,s,=1;i\neq r;j\neq s}^{d+1}\delta_{i-r,j-s}E'_{(i-1)(d+m)+j,(r-1)(d+m)+s}\\
&+d\sum_{i,r=1}^{d+1}E'_{(i-1)(d+m)+i,(r-1)(d+m)+r}+(d-1)\sum_{i,j,r,s,=1;i\neq j;r\neq s}^{d+1}\delta_{i-j,r-s}E'_{(i-1)(d+m)+j,(r-1)(d+m)+s}\\
&+d\sum_{i=1;r=d+2}^{d+1;d+m}E'_{(i-1)(d+m)+i,(r-1)(d+m)+r}+(d-1)\sum_{i,j=1;r=d+2;i\neq j}^{d+1,d+m}E'_{(i-1)(d+m)+j,(r-1)(d+m)+r}\\
&+d\sum_{i=d+2,r=1}^{d+m;d+1}E'_{(i-1)(d+m)+i,(r-1)(d+m)+r}+(d-1)\sum_{i=d+2;r,s=1;r\neq s}^{d+m;d+1}E'_{(i-1)(d+m)+i,(r-1)(d+m)+s}\\
&+d^2\sum_{i,j=d+2}^{d+m}E'_{(i-1)(d+m)+j,(i-1)(d+m)+j}+d^2\sum_{i,r=d+2}^{d+m}E'_{(i-1)(d+m)+i,(r-1)(d+m)+r}\\
&+d\sum_{i=1;j=d+2}^{d+1;d+m}E'_{(i-1)(d+m)+j,(i-1)(d+m)+j}+(d-1)\sum_{i,r=1;j=2;i\neq r}^{d+1,d+m}E'_{(i-1)(d+m)+j,(r-1)(d+m)+j}\\
&+d\sum_{i=d+2;j=1}^{d+m;d+1}E'_{(i-1)(d+m)+j,(i-1)(d+m)+j}+(d-1)\sum_{i=d+2;j,s=1;j\neq s}^{d+m;d+1}E'_{(i-1)(d+m)+j,(i-1)(d+m)+s}\\
\end{aligned}
\end{equation*}

\begin{equation*}
	\begin{aligned}
&=\sum_{i,j=1}^{d+1}E'_{(i-1)(d+m)+j,(i-1)(d+m)+j}+2(d-1)\sum_{i,j,r,s=1}^{d+1}\delta_{i-r,j-s}E'_{(i-1)(d+m)+j,(r-1)(d+m)+s}\\
&+\sum_{i,r=1}^{d+1}E'_{(i-1)(d+m)+i,(r-1)(d+m)+r}
+\sum_{i=1;r=d+2}^{d+1;d+m}E'_{(i-1)(d+m)+i,(r-1)(d+m)+r}\\
&+(d-1)\sum_{i,j=1;r=d+2 }^{d+1,d+m}E'_{(i-1)(d+m)+j,(r-1)(d+m)+r}+
\sum_{i=d+2,r=1}^{d+m;d+1}E'_{(i-1)(d+m)+i,(r-1)(d+m)+r}\\
&+(d-1)\sum_{i=d+2;r,s=1 }^{d+m;d+1}E'_{(i-1)(d+m)+i,(r-1)(d+m)+s}+
d^2\sum_{i,j=d+2}^{d+m}E'_{(i-1)(d+m)+j,(i-1)(d+m)+j}\\
&+d^2\sum_{i,r=d+2}^{d+m}E'_{(i-1)(d+m)+i,(r-1)(d+m)+r}+
\sum_{i=1;j=d+2}^{d+1;d+m}E'_{(i-1)(d+m)+j,(i-1)(d+m)+j}\\
&+(d-1)\sum_{i,r=1;j=2; }^{d+1,d+m}E'_{(i-1)(d+m)+j,(r-1)(d+m)+j}+
\sum_{i=d+2;j=1}^{d+m;d+1}E'_{(i-1)(d+m)+j,(i-1)(d+m)+j}\\
&+(d-1)\sum_{i=d+2;j,s=1 }^{d+m;d+1}E'_{(i-1)(d+m)+j,(i-1)(d+m)+s}
\end{aligned}
\end{equation*}

After simplification we will get
\begin{equation*}
	\begin{aligned}
		M &=\sum_{x,y=1}^{d+m}(E_{x,x}\otimes E_{y,y}+E_{x,y}\otimes E_{x,y})+(d^2-1)\sum_{r,s=d+2}^{d+m}(E_{r,s}\otimes E_{r,s}+E_{r,r}\otimes E_{s,s})\,+\\
		&2(d-1)\sum_{i,j=1}^{d+1}E_{i,j}\otimes S^{Mod[i-j-1,d+1]+1}
		+(d-1)\sum_{r=d+2}^{d+m}\sum_{i,j=1}^{d+1}(E_{i,r}\otimes E_{j,r}
		+E_{r,i}\otimes E_{r,j}
		+E_{i,j}\otimes E_{r,r}+E_{r,r}\otimes E_{i,j})
	\end{aligned}	
\end{equation*}
where $ S=\sum_{k=1}^{d+1}\ket{e_{\Mod[k,d+1]+1}}\bra{k}$.

\begin{lem}\label{$rho^T$}
If $\rho$ is the choi state corresponding to completely positive map $\Phi$ in (\ref{C.P.E}) then $\rho$ is PPT.
\end{lem}
\begin{proof}

\begin{equation}
	\Phi(X)=\frac{1}{d(d+m)}\sum_{i=1}^{d+m}V_i^\dagger XV_i
\end{equation}
where 

\begin{equation}
	V_i=\sum_{k=1}^{d}\ket {e^d_k}\bra{e^{d+m}_{k+i-1}}, \quad\quad i=1,2...d+1
\end{equation}

\begin{equation}
	V_i=\ket{\J}\bra{e^{d+m}_i}, \quad\quad  i=d+2,.....d+m 
\end{equation}
So
\[\Phi(X)=\frac{1}{d(d+m)}\left(\sum_{i=1}^{d+1}\sum_{k,l=1}^{d} X_{k,l}\ket{e_{k+i-1}}\bra{e_{l+i-1}}+\sum_{i=d+2}^{m}\ket{e_i}\bra{\J}X\ket{\J}\bra{e_i}\right)\]
Since $\rho$ is the Choi- state corresponding to $\Phi$ so
\[\rho=\frac{1}{d(d+m)}\left(\sum_{r,s=1}^d \ket{e_r^d}\bra{e_s^d}\otimes \left(\sum_i^{d+1} \ket{e_{\Mod[r+i-2,d+1]+1}^{d+m}}\bra{e_{\Mod[s+i-2,d+1]+1}^{d+m}} +(0\oplus I_{m-1})\right)\right)\]
\[=\frac{1}{d(d+m)}\left(\sum_{i=1}^{d+1}\sum_{r,s=1}^d\ket{e_r^d}\bra{e_s^d}\otimes\ket{e_{\Mod[r+i-2]+1}^{d+m}}\bra{e_{\Mod[s+i-2,d+1]+1}^{d+m}} +\J\otimes(0\oplus I_{m-1})\right)\]
\[=\frac{1}{d(d+m)}\left(\sum_{i=1}^{d+1}\sum_{r,s=1}^d  \ket{e_r^de_{r+i-1}^{d+m}}\bra{e_s^de_{s+i-1}^{d+m}} +\J\otimes(0\oplus I_{m-1})\right)\]

\begin{equation}{\label{PS}}
\rho^{T_1}=\frac{1}{d(d+m)}\left(\sum_{i=1}^{d+1}\sum_{r,s=1}^d\ket{e_s^d}\bra{e_r^d}\otimes \ket{e_{\Mod[r+i-2,d+1]+1}^{d+m}}\bra{e_{\Mod[s+i-2,d+1]+1}^{d+m}} +\J\otimes(0\oplus I_{m-1})\right)
	\end{equation}
Consider 
\[S=\sum_{i=1}^{d+1} \ket{e^{d+m}_{i+1}}\bra{e^{d+m}_i}.\]
So
\[\sum_{i=1}^{d+1}\ket{e_{\Mod[r+i-2,d+1]+1}^{d+m}}\bra{e_{\Mod[s+i-2,d+1]+1}^{d+m}}=S^{\Mod[r-s-1,d+1]+1}\]
then
\[\sum_{i=1}^{d+1}\sum_{r,s=1}^d\ket{e_s^d}\bra{e_r^d}\otimes \ket{e_{\Mod[r+i-2,d+1]+1}^{d+m}}\bra{e_{\Mod[s+i-2,d+1]+1}^{d+m}}=\begin{pmatrix}
	\idm  \\  
	S^\dagger \\
	 S^{2^\dagger}\\
	.\\
	.\\
	S^{d^\dagger}
\end{pmatrix}\begin{pmatrix}
	\idm & S& S^2  &.  &.  &S^d
\end{pmatrix}.\]
$\rho^{T_1}\geq 0$ as both term in equation(\ref{PS}) is positive semidefinite. Hence $\rho$ is PPT.
\end{proof}

	\printbibliography
	\end{document}